\documentclass[english,onecolumn]{IEEEtran}
\usepackage[T1]{fontenc}
\usepackage[latin9]{inputenc}
\usepackage{amsthm}
\usepackage{amsmath}
\usepackage{graphicx}
\usepackage{amssymb}
\usepackage{esint}

\makeatletter
\theoremstyle{plain}
\newtheorem{thm}{Theorem}
  \theoremstyle{plain}
  \newtheorem{lem}[thm]{Lemma}

\IEEEoverridecommandlockouts

\makeatother

\usepackage{babel}

\begin{document}
\title{Ensemble estimation of multivariate $f$-divergence}
\author{\IEEEauthorblockN{Kevin R. Moon and Alfred O. Hero III}\\ \IEEEauthorblockA{Dept. of EECS, University of Michigan, Ann Arbor, Michigan \\ Email: \{krmoon, hero\}@umich.edu }\thanks{This work was partially supported by NSF grant CCF-1217880 and a NSF Graduate Research Fellowship to the first author under Grant No. F031543.}}

\maketitle
\newcommand{\fhat}[1]{\hat{\mathbf{f}}_{#1,k_#1}}
\newcommand{\bz}[1]{b_{#1,k_2}}

\global\long\def\gtay#1#2#3{\mathbf{#1}_{#2}^{(#3)}}
\global\long\def\gk{\hat{\mathbf{G}}_{k_{1},k_{2}}}

\global\long\def\ekl#1{\hat{\mathbf{F}}_{k(#1)}}
\global\long\def\ek{\hat{\mathbf{F}}_{k_{1},k_{2}}}

\global\long\def\ehatl#1#2{\hat{\mathbf{e}}_{#1,k(#2)}}
\global\long\def\ehat#1{\hat{\mathbf{e}}_{#1,k_{#1}}}

\global\long\def\lhat{\hat{\mathbf{L}}_{k_{1},k_{2}}}
\global\long\def\lhatl#1{\hat{\mathbf{L}}_{k(#1)}}

\global\long\def\bE{\mathbb{E}}
\global\long\def\ez{\mathbb{E}_{\mathbf{Z}}}
\global\long\def\var{\mathbb{V}}
\global\long\def\bias{\mathbb{B}}

\begin{abstract}
$f$-divergence estimation is an important problem in the fields of
information theory, machine learning, and statistics. While several
divergence estimators exist, relatively few of their convergence rates
are known. We derive the MSE convergence rate for a density plug-in
estimator of $f$-divergence. Then by applying the theory of optimally
weighted ensemble estimation, we derive a divergence estimator with
a convergence rate of $O\left(\frac{1}{T}\right)$ that is simple
to implement and performs well in high dimensions. We validate our
theoretical results with experiments.
\end{abstract}

\section{Introduction}

$f$-divergence is a measure of the difference between distributions
and is important to the fields of information theory, machine learning,
and statistics~\cite{csiszar1967information}. Many different kinds
of $f$-divergences have been defined including the Kullback-Leibler
(KL)~\cite{kullback1951divergence} and R\'{e}nyi-$\alpha$~\cite{renyi1961divergence}.
A special case of the KL divergence is mutual information which gives
the capacities in data compression and channel coding~\cite{cover2006infotheory}.
Mutual information estimation has also been used in applications such
as feature selection~\cite{peng2005feature}, fMRI data processing~\cite{chai2009fmri},
and clustering~\cite{lewi2006cluster}. Entropy is also a special
case of divergence where one of the distributions is the uniform distribution.
Entropy estimation is useful for intrinsic dimension estimation~\cite{carter2010dim},
texture classification and image registration~\cite{hero2002graphs},
and many other applications. Additionally, divergence estimation is
useful for empirically estimating the decay rates of error probabilities
of hypothesis testing~\cite{cover2006infotheory} and extending machine
learning algorithms to distributional features~\cite{poczos2011divergence,oliva2013distribution}.
For other applications of divergence estimation, see~\cite{wang2009divergence}. 

We consider the problem of estimating the $f$-divergence when only
two finite populations of independent and identically distributed
(i.i.d.) samples are available from some unknown, nonparametric, smooth,
$d$-dimensional distributions. While several estimators of divergence
have been previously defined, the convergence rates are known for
only a few of them. Our first contribution is to derive convergence
rates for kernel density plug-in $f$-divergence estimators with an
adaptive $k$-nearest neighbor ($k$-nn) kernel. Our second contribution
is to extend the theory of optimally weighted ensemble entropy estimation
developed in~\cite{sricharan2013ensemble} to obtain a divergence
estimator with a convergence rate of $O\left(\frac{1}{T}\right)$
where $T$ is the sample size. This is accomplished by solving an
offline convex optimization problem.

\subsection{Related Work}

Several estimators for some $f$-divergences already exist. For example,
P\'{o}czos \& Schneider~\cite{poczos2011divergence} established
weak consistency of a bias-corrected $k$-nn estimator for R\'{e}nyi-$\alpha$
and other divergences of similar form. Wang et al~\cite{wang2009divergence}
gave an estimator for the KL divergence. Other mutual information
and divergence estimators based on plug-in histogram schemes have
been proven to be consistent~\cite{darbellay1999MIest,wang2005divergencepart,silva2010partition,le2013partition}.
However none of these works studied the convergence rates of their
estimators while our ensemble approach requires an explicit expression
of the asymptotic bias and variance. Hero et al~\cite{hero2002graphs}
provided an estimator for R\'{e}nyi-$\alpha$ divergence but assumed
that one of the densities was known. 

Nguyen et al~\cite{nguyen2010divergence} proposed a method for estimating
$f$-divergences by estimating the likelihood ratio of the two densities
by solving a convex optimization problem and then plugging it into
the divergence formulas. For this method they prove that the minimax
convergence rate is parametric ($O\left(\frac{1}{T}\right)$) when
the likelihood ratio is in the bounded H\"{o}lder class $\Sigma_{\kappa}(\beta,L,r)$
with $\beta\geq d/2$. This assumption is weaker than ours which requires
the densities to be at least $d$ times differentiable. However, solving
the convex problem of~\cite{nguyen2010divergence} is similar in
complexity to training the SVM (between $O(T^{2})$ and $O(T^{3})$)
which can be demanding when $T$ is very large. In contrast, our method
of optimally weighted ensemble estimation depends only on simple density
plug-in estimates and an offline convex optimization problem. Thus
the most computationally demanding step in our approach is the calculation
of the $k$-nn distances which has complexity no greater than $O(T^{2})$.

Singh and P\'{o}czos~\cite{singh2014exp} provided an estimator
for R\'{e}nyi-$\alpha$ divergences that uses a {}``mirror image''
kernel density estimator. They prove a convergence rate of $O\left(\frac{1}{T}\right)$
when $\beta\geq d$ for each of the densities. However this method
requires several computations at each boundary of the support of the
densities which becomes difficult to implement as $d$ gets large.
Also, this method requires knowledge of the support of the densities
which may not be possible for some problems.

The main results of our paper are as follows. First, under the assumption
that the densities are smooth, lower bounded, and have bounded support,
the mean squared error (MSE) of a kernel density plug-in estimator
of $f$-divergence converges to zero at the non-parametric rate of
$O\left(T^{-1/d}\right),$ which becomes exceedingly slow as dimension
$d$ increases. Second, the proposed weighted ensemble estimator of
divergence is simple to implement and its MSE converges at the parametric
rate of $O\left(\frac{1}{T}\right)$. Third, the proposed estimator
of divergence is shown by simulation to outperform standard kernel
density plug-in estimators for modest sample sizes ($T\geq400$) and
in high dimensions ($d\geq4$). Finally, the proposed divergence estimator
performs well even for densities with unbounded support (Gaussian),
suggesting that our theory holds under significantly weaker assumptions.

\subsection{Organization and Notation}

The paper is organized as follows. Section~\ref{sec:Weighted_ensemble}
provides the theory underlying the optimally weighted ensemble estimator.
Section~\ref{sec:div_application} applies this theory to $f$-divergence
estimation and gives convergence results for the estimators, while
Section~\ref{sec:Proofs} provides proofs. Section~\ref{sec:Experiments}
gives some experimental results that illustrate the performance of
our estimators as a function of $T$ and $d$. Section~\ref{sec:Conclusion}
concludes the paper.

Bold face type is used for random variables and random vectors. Let
$f_{1}$ and $f_{2}$ be densities and define $L(x)=\frac{f_{1}(x)}{f_{2}(x)}$.
The conditional expectation given a random variable $\mathbf{Z}$
is denoted $\mathbb{E}_{\mathbf{Z}}$. The variance of a random variable
is denoted $\var$ and the bias of an estimator is denoted $\bias$.

\section{Weighted ensemble estimation\label{sec:Weighted_ensemble}}

Let $\bar{l}=\left\{ l_{1},\dots,l_{L}\right\} $ be a set of index
values and $T$ the number of samples available. For an indexed ensemble
of estimators $\left\{ \hat{\mathbf{E}}_{l}\right\} _{l\in\bar{l}}$
of the parameter $E$, the weighted ensemble estimator with weights
$w=\left\{ w\left(l_{1}\right),\dots,w\left(l_{L}\right)\right\} $
satisfying $\sum_{l\in\bar{l}}w(l)=1$ is defined as \[
\hat{\mathbf{E}}_{w}=\sum_{l\in\bar{l}}w\left(l\right)\hat{\mathbf{E}}_{l}.\]
$\hat{\mathbf{E}}_{w}$ is asyptotically unbiased if the estimators
$\left\{ \hat{\mathbf{E}}_{l}\right\} _{l\in\bar{l}}$ are asymptotically
unbiased. Typically the MSE of a plug-in estimator is dominated by
the bias. The key idea to reducing MSE is that by choosing appropriate
weights $w$, we can greatly decrease the bias in exchange for some
increase in variance. Suppose the following conditions are satisfied
by $\left\{ \hat{\mathbf{E}}_{l}\right\} _{l\in\bar{l}}$~\cite{sricharan2013ensemble}:
\begin{itemize}
\item $\mathcal{C}.1$ The bias is given by \[
\bias\left(\hat{\mathbf{E}}_{l}\right)=\sum_{i\in J}c_{i}\psi_{i}(l)T^{-i/2d}+O\left(\frac{1}{\sqrt{T}}\right),\]
 where $c_{i}$ are constants depending on the underlying density,
$J=\left\{ i_{1},\dots,i_{I}\right\} $ is a finite index set with
$I<L$, $\min(J)>0$ and $\max(J)\leq d$, and $\psi_{i}(l)$ are
basis functions depending only on the parameter $l$. 
\item $\mathcal{C}.2$ The variance is given by \[
\var\left[\hat{\mathbf{E}}_{l}\right]=c_{v}\left(\frac{1}{T}\right)+o\left(\frac{1}{T}\right).\]
\end{itemize}
\begin{thm}
\cite{sricharan2013ensemble} \label{thm:ensemble}Assume conditions
$\mathcal{C}.1$ and $\mathcal{C}.2$ hold for an ensemble of estimators
$\left\{ \hat{\mathbf{E}}_{l}\right\} _{l\in\bar{l}}$. Then there
exists a weight vector $w_{0}$ such that \[
\mathbb{E}\left[\left(\hat{\mathbf{E}}_{w_{0}}-E\right)^{2}\right]=O\left(\frac{1}{T}\right).\]
The weight vector $w_{0}$ is the solution to the following convex
optimization problem:\[
\begin{array}{rl}
\min_{w} & ||w||_{2}\\
subject\, to & \sum_{l\in\bar{l}}w(l)=1,\\
 & \gamma_{w}(i)=\sum_{l\in\bar{l}}w(l)\psi_{i}(l)=0,\, i\in J.\end{array}\]

\end{thm}

\section{Application to divergence estimation\label{sec:div_application}}

Theorem~\ref{thm:ensemble} was applied in~\cite{sricharan2013ensemble}
to obtain an entropy estimator with parametric convergence rates $O\left(\frac{1}{T}\right).$
An analogous theorem will be presented that applies ensemble estimation
of estimators of $f$-divergence. Specifically, we focus on divergences
that include the form~\cite{csiszar1967information} \begin{equation}
G(f_{1},f_{2})=\int g\left(\frac{f_{1}(x)}{f_{2}(x)}\right)f_{2}(x)dx,\label{eq:fdivergences}\end{equation}
for some smooth, convex function $g(f)$. Divergences that have this
form include the Renyi divergence ($g(x)=x^{\alpha}$) and the KL
divergence ($g(x)=-\ln x$). We assume that the $d$-dimensional multivariate
densities $f_{1}$ and $f_{2}$ have finite support $\mathcal{S}=\left[a,b\right]^{d}$.
Assume that $T=N+M_{2}$ i.i.d. realizations $\left\{ \mathbf{X}_{1},\dots,\mathbf{X}_{N},\mathbf{X}_{N+1},\dots,\mathbf{X}_{N+M_{2}}\right\} $
are available from the density $f_{2}$ and $M_{1}$ i.i.d. realizations
$\left\{ \mathbf{Y}_{1},\dots,\mathbf{Y}_{M_{1}}\right\} $ are available
from the density $f_{1}$.

We use $k$-nn density estimators in our proposed $f$-divergence
estimator. Assume that $k_{i}\leq M_{i}.$ Let $\rho_{2,k_{2}}(i)$
be the distance of the $k_{2}$th nearest neighbor of $X_{i}$ in
$\left\{ X_{N+1},\dots,X_{T}\right\} $ and let $\rho_{1,k_{1}}(i)$
be the distance of the $k_{1}$th nearest neighbor of $X_{i}$ in
$\left\{ Y_{1},\dots,Y_{M_{1}}\right\} .$ Then the $k$-nn density
estimate is~\cite{loftsgaarden1965knn}\[
\fhat{i}(X_{j})=\frac{k_{i}}{M_{i}\bar{c}\mathbf{\mathbf{\rho}}_{i,k_{i}}^{d}(j)},\]
where $\bar{c}$ is the volume of a $d$-dimensional unit ball.

The plug-in estimator of divergence is constructed similarly to~\cite{sricharan2013ensemble}.
The data from $f_{2}$ are randomly divided into two parts $\left\{ \mathbf{X}_{1},\dots,\mathbf{X}_{N}\right\} $
and $\left\{ \mathbf{X}_{N+1},\dots,\mathbf{X}_{N+M_{2}}\right\} $.
The density estimate $\fhat{2}$ is found at the $N$ points $\left\{ \mathbf{X}_{1},\dots,\mathbf{X}_{N}\right\} $
using the $M_{2}$ realizations $\left\{ \mathbf{X}_{N+1},\dots,\mathbf{X}_{N+M_{2}}\right\} $.
Splitting the data in this manner is a common approach to debiasing
and variance reduction in non-parametric estimation. Similarly, the
density estimate $\fhat{1}$ is found at the $N$ points $\left\{ \mathbf{X}_{1},\dots,\mathbf{X}_{N}\right\} $
using the $M_{1}$ realizations $\left\{ \mathbf{Y}_{1},\dots,\mathbf{Y}_{M_{1}}\right\} $.
Define $\lhat(x)=\frac{\fhat{1}(x)}{\fhat{2}(x)}.$ The functional
$G(f_{1},f_{2})$ is then approximated as \begin{equation}
\gk=\frac{1}{N}\sum_{i=1}^{N}g\left(\lhat\left(\mathbf{X}_{i}\right)\right).\label{eq:estimator}\end{equation}
This is a plug-in estimator in the sense that we plug in the estimates
to the argument of the expectation, and then use the empirical average
to calculate the expectation.

Similar to~\cite{sricharan2013ensemble}, the principal assumptions
we make on the densities $f_{1}$ and $f_{2}$ and the functional
$g$ are that: 1) $f_{1}$, $f_{2},$ and $g$ are smooth; 2) $f_{1}$
and $f_{2}$ have common bounded support sets $\mathcal{S}$; 3) $f_{1}$
and $f_{2}$ are strictly lower bounded. Specifically: 
\begin{itemize}
\item $(\mathcal{A}.0)$: Assume that $k_{i}=k_{0}M_{i}^{\beta}$ with $0<\beta<1$,
that $M_{2}=\alpha_{frac}T$ with $0<\alpha_{frac}<1$. 
\item $(\mathcal{A}.1)$: Assume there exist constants $\epsilon_{0},\epsilon_{\infty}$
such that $0<\epsilon_{0}\leq f_{i}(x)\leq\epsilon_{\infty}<\infty,\,\forall x\in S.$ 
\item $(\mathcal{A}.2)$: Assume that the densities $f_{i}$ have continuous
partial derivatives of order $d$ in the interior of $\mathcal{S}$
that are upper bounded. 
\item $(\mathcal{A}.3)$: Assume that $g$ has derivatives $g^{(j)}$ of
order $j=1,\dots,\max\{\lambda,d\}$ where $\lambda\beta>1$. 
\item $(\mathcal{A}.4$): Assume that $\left|g^{(j)}\left(f_{1}(x)/f_{2}(x)\right)\right|$,
$j=0,\ldots,\max\{\lambda,d\}$ are strictly upper bounded for $\epsilon_{0}\leq f_{i}(x)\leq\epsilon_{\infty}.$ 
\item $(\mathcal{A}.5)$: Let $\epsilon\in(0,1)$, $\delta\in(2/3,1)$,
and $\mathcal{C}(k)=\exp\left(-3k^{(1-\delta)}\right).$ For fixed
$\epsilon,$ define $p_{l,i}=(1-\epsilon)\epsilon_{0}\frac{k_{i}-1}{M_{i}}$,
$p_{u,i}=(1+\epsilon)\epsilon_{\infty}\frac{k_{i}-1}{M_{i}}$, $q_{l,i}=\frac{k_{i}-1}{M_{i}\bar{c}D^{d}},$
and $q_{u,i}=(1+\epsilon)\epsilon_{\infty}$ where $D$ is the diameter
of the support $S.$ Let $\mathbf{P}_{i}$ be a beta distributed random
variable with parameters $k_{i}$ and $M_{i}-k_{i}+1.$ Define $p_{l}=\frac{p_{l,1}}{p_{u,2}}$
and $p_{u}=\frac{p_{u,1}}{p_{l,2}}$. Assume that for $U(L)=g(L),\, g^{(3)}(L),$
and $g^{(\lambda)}(L),$

\begin{itemize}
\item $(i)\,\mathbb{E}\left[\sup_{L\in(p_{l},p_{u})}\left|U\left(L\frac{\mathbf{P}_{2}}{\mathbf{P}_{1}}\right)\right|\right]=G_{1}<\infty$, 
\item $(ii)\,\sup_{L\in\left(\frac{q_{l,1}}{q_{u,2}},\frac{q_{u,1}}{q_{l,2}}\right)}\left|U\left(L\right)\right|\mathcal{C}\left(k_{1}\right)\mathcal{C}\left(k_{2}\right)=G_{2}<\infty,$ 
\item $(iii)\,\mathbb{E}\left[\sup_{L\in\left(\frac{q_{l,1}}{p_{u,2}},\frac{q_{u,1}}{p_{l,2}}\right)}\left|U\left(L\mathbf{P}_{2}\right)\right|\mathcal{C}\left(k_{1}\right)\right]=G_{3}<\infty,$ 
\item $(iv)\,\mathbb{E}\left[\sup_{L\in\left(\frac{p_{l,1}}{q_{u,2}},\frac{p_{u,1}}{q_{l,2}}\right)}\left|U\left(\frac{L}{\mathbf{P}_{1}}\right)\right|\mathcal{C}\left(k_{2}\right)\right]=G_{4}<\infty,\,\forall M_{i}.$ 
\end{itemize}
\end{itemize}
Densities for which assumptions $\mathcal{A}.0-\mathcal{A}.5$ hold
include the truncated Gaussian distribution and the Beta distribution
on the unit cube. Functions for which the assumptions hold include
$g(L)=-\ln L$ and $g(L)=L^{\alpha}.$

\subsection{Analysis of mean squared error}

The following hold under assumptions $\mathcal{A}.0-\mathcal{A}.5$:
\begin{thm}
\label{thm:bias}The bias of the plug-in estimator $\gk$ is given
by\[
\bias\left(\gk\right)=\sum_{j=1}^{d}\left(c_{6,j,1}\left(\frac{k_{1}}{M_{1}}\right)^{\frac{j}{d}}+c_{6,j,2}\left(\frac{k_{2}}{M_{2}}\right)^{\frac{j}{d}}\right)+\left(c_{4,1}+c_{4,2}+c_{6,3}\right)\left(\frac{1}{k_{2}}\right)\]
\[
+c_{4,3}\left(\frac{1}{k_{1}}\right)+o\left(\frac{1}{k_{1}}+\frac{1}{k_{2}}+\frac{k_{1}}{M_{1}}+\frac{k_{2}}{M_{2}}\right).\]

\end{thm}
Figure~\ref{fig:heatmap} gives a heatmap showing the leading term
$O\left(\left(\frac{k}{M}\right)^{1/d}\right)$ as a function of $d$
and $M$. 

\begin{figure}
\centering

\includegraphics[width=0.4\textwidth]{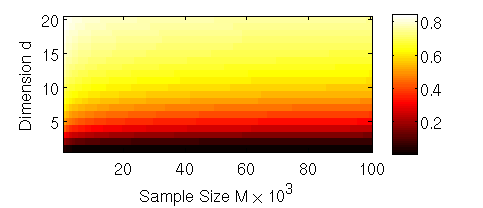}

\caption{Heat map of predicted bias of non-averaged $f$-divergence estimator
based on Theorem~\ref{thm:bias} as a function of dimension and sample
size. Note the phase transition as dimension $d$ increases for fixed
sample size $M$: bias remains small only for relatively small values
of $d.$ The proposed weighted ensemble averaged estimator removes
this phase transition when the densities are sufficiently smooth.
\label{fig:heatmap}}
\end{figure}

\begin{thm}
\label{thm:variance}The variance of the plug-in estimator $\gk$
is \[
\var\left[\gk\right]=c_{9}\left(\frac{1}{N}\right)+c_{8,1}\left(\frac{1}{M_{1}}\right)+c_{8,2}\left(\frac{1}{M_{2}}\right)+o\left(\frac{1}{M_{1}}+\frac{1}{M_{2}}+\frac{1}{N}+\frac{1}{k_{1}^{2}}+\frac{1}{k_{2}^{2}}\right).\]

\end{thm}
Note that the constants in front of the terms that depend on $k_{i}$
and $M_{i}$ are not identical for different $i$. However, these
constants depend on the densities $f_{1}$ and $f_{2}$ which are
often unknown and thus impossible to compute in practice. The rates
given here are very similar to the rates derived for the entropy plug-in
estimator in~\cite{sricharan2013ensemble}. The differences are in
the constants in front of the rates, the dependence on the number
of samples from two distributions instead of one, and the $o\left(\frac{1}{k_{i}^{2}}\right)$
terms in the expression for the variance. The key to reducing mean
squared error (MSE) is that by applying Theorem~\ref{thm:ensemble},
the dependence of the MSE on $d$ will be greatly reduced.

\subsection{Weighted ensemble divergence estimator}

\label{sub:optimal_Est}Let $L>I=d-1$ and choose $\bar{l}=\left\{ l_{1},\dots,l_{L}\right\} $
to be positive real numbers. Assume that $M_{1}=O\left(M_{2}\right).$
Let $k(l)=l\sqrt{M_{2}}$, $\hat{\mathbf{G}}_{k(l)}:=\hat{\mathbf{G}}_{k(l),k(l)},$
and $\hat{\mathbf{G}}_{w}:=\sum_{l\in\bar{l}}w(l)\hat{\mathbf{G}}_{k(l)}.$
From Theorems~\ref{thm:bias} and~\ref{thm:variance}, the biases
of the ensemble estimators $\left\{ \hat{\mathbf{G}}_{k(l)}\right\} _{l\in\bar{l}}$
satisfy the condition $\mathcal{C}.1$ when $\psi_{i}(l)=l^{i/d}$
and $J=\{1,\dots,d-1\}$ since \[
\bias\left(\hat{\mathbf{G}}_{k(l)}\right)=\sum_{j=1}^{d-1}O\left(l^{j/d}M_{2}^{-\frac{j}{2d}}\right)+O\left(\frac{1}{\sqrt{M_{2}}}\right).\]
 The general form of the variance of $\hat{\mathbf{G}}_{k(l)}$ also
follows $\mathcal{C}.2$ since $N,\, M_{2}=\Theta(T)$ (see $\mathcal{A}.0$).
Thus we can find the optimal weight $w_{0}$ by using Theorem~\ref{thm:ensemble}
to obtain a plug-in $f$-divergence estimator with convergence rate
of $O\left(\frac{1}{T}\right).$

\subsection{Proofs of Theorems~\ref{thm:bias} and~\ref{thm:variance}}

\label{sec:Proofs}Like for the case of entropy estimation studied
in~\cite{sricharan2013ensemble}, the principal tools for the proofs
of Theorems~\ref{thm:bias} and~\ref{thm:variance} are concentration
inequalities and moment bounds applied to a higher order Taylor expansion
of the functional (\ref{eq:estimator}). However, as the functional
(\ref{eq:estimator}) depends on the ratio of densities, the analysis
is more complicated than that of~\cite{sricharan2013ensemble} since
we have to bound the covariances between products of $\ehat 1(\mathbf{Z})$
and $\ehat 2(\mathbf{Z})$ where $\mathbf{Z}$ is drawn from $f_{2}$,
$\ehat i(\mathbf{Z})=\fhat{i}(\mathbf{Z})-\mathbb{E}_{\mathbf{Z}}\fhat{i}(\mathbf{Z})$,
and $\ek(\mathbf{Z})=\lhat(\mathbf{Z})-\mathbb{E}_{\mathbf{Z}}\left(\lhat\right)$.
Using Lemmas~5, 8, and 9 in~\cite{sricharan2013ensemble}, modified
for application to $f_{1}$ and $f_{2}$, and two new Lemmas (Lemma~\ref{lem:ekhat}
and Lemma~\ref{lem:cov_ekhat} in the appendices) will establish
Theorem~\ref{thm:bias} and Theorem~\ref{thm:variance}. The modified
versions of Lemmas~5 and 8 from~\cite{sricharan2013ensemble} are
given in Lemma~\ref{lem:ekhat} and Lemma~\ref{lem:cov_ekhat},
respectively while the modified version of Lemma~9 from~\cite{sricharan2013ensemble}
is given as Lemma~\ref{lem:bounded}. The details are given in Appendix~\ref{sec:bias}
and Appendix~\ref{sec:variance}.

\section{Experiments\label{sec:Experiments}}

To demonstrate the accuracy of the theoretical predictions of the
performance of the ensemble method, we estimated the R\'{e}nyi $\alpha$-divergence
between two truncated normal densities with varying dimension and
sample size restricted to the unit cube. The densities have means
$\bar{\mu}_{1}=0.7*\bar{1}_{d}$, $\bar{\mu}_{2}=0.3*\bar{1}_{d}$
and covariance matrices $\sigma_{i}I_{d}$ where $\sigma_{1}=0.1,$
$\sigma_{2}=0.3,$ $\bar{1}_{d}$ is a $d$-dimensional vector of
ones, and $I_{d}$ is a $d$-dimensional identity matrix. We used
$\alpha=0.8$ and computed the estimates for the truncated kernel
density plug-in estimate, the $k$-nn plug-in estimate, and the optimally
weighted $k$-nn estimate. Since we have a finite number of samples,
we obtain $w_{0}$ by solving the second convex optimization problem
in~\cite{sricharan2013ensemble} which introduces a slack variable
on the bias constraint to better control the variance.

\begin{figure}
\centering

\includegraphics[width=0.25\textwidth]{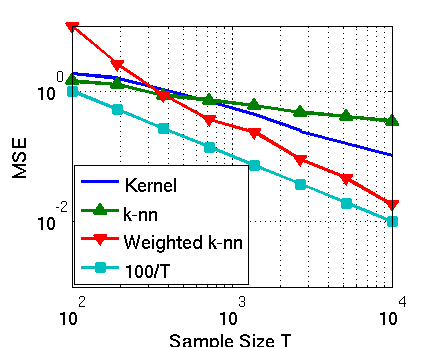}\includegraphics[width=0.25\textwidth]{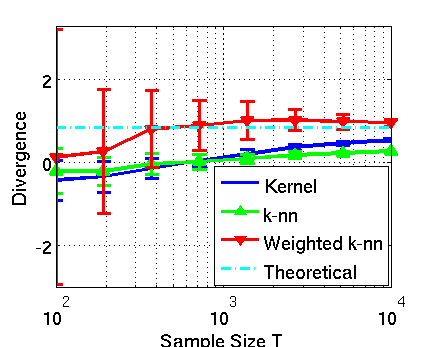}

\caption{(Left) Log-log plot of MSE of the truncated uniform kernel and $k$-nn
plug-in estimators ({}``Kernel'', {}``$k$-nn''), our proposed
weighted ensemble estimator, and the theoretical bound from Theorem~\ref{thm:ensemble}
scaled by a constant $(100/T)$. (Right) Average estimated divergence
for each estimator with error bars indicating the standard deviation.
Estimates for both plots are calculated from 100 trials for various
sample sizes with fixed $d=5$. The proposed estimator outperforms
the others for $T>400$ and is less biased. \label{fig:trunc_mse}}
\end{figure}

The left plot in Fig.~\ref{fig:trunc_mse} shows the MSE of all three
estimators for various sample sizes and fixed $d=5$. This experiment
shows that the optimally weighted $k$-nn estimate consistently outperforms
the others for sample sizes greater than $400$. The slope of the
MSE of the optimally weighted $k$-nn estimate also matches the slope
of the theoretical bound well.

The right plot in Fig.~\ref{fig:trunc_mse} shows the corresponding
average estimated divergence and standard deviation for the three
estimates. From the plot, the bias is consistently lowest for the
ensemble estimate while the variance is highest suggesting that bias
is decreased at the expense of increased variance. 

We repeated the experiment with a fixed sample size of $T=3000$ and
varying dimension. Based on the MSE, the ensemble estimate does better
than the other methods for $d\geq4$ and is comparable to the other
methods for $d<4$ (see Fig.~\ref{fig:trunc_increasingd}). Note
that MSE appears to increase slightly for all estimators as $d$ increases.
This is likely due to the dependence of the constants in the bias
and variance terms on the densities and because we are using a fixed
number of estimators $L$~\cite{sricharan2013ensemble}.

\begin{figure}
\centering

\includegraphics[width=0.35\textwidth]{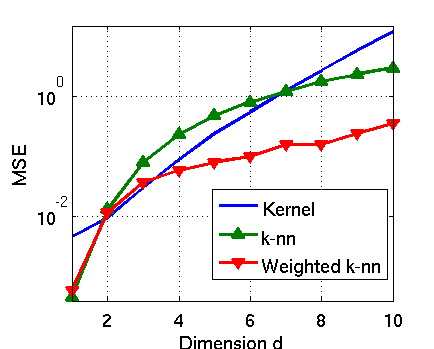}\caption{Plot of MSE of 100 trials of the estimators for various dimensions
at a fixed sample size $T=3000$. The proposed estimator outperforms
the others for $d\geq4$ and performs similarly for $d\leq3.$\label{fig:trunc_increasingd}}

\end{figure}

To test the limits of our theoretical results, we also ran the experiment
for non-truncated Gaussian random variables. Figure~\ref{fig:notrunc_mse_n}
shows the MSE as a function of sample size and dimension, respectively.
For fixed $d=5,$ the weighted ensemble estimate has the lowest MSE
for almost all sample sizes in the range considered. For fixed $T=3000,$
the MSE of the kernel plug-in method stays low for small dimension
but then rapidly increases as $d$ increases. For the weighted $k$-nn
method, the MSE increases at a slower rate as $d$ increases and is
lowest for $d\leq2$ and $d\geq5.$

\begin{figure}
\centering

\includegraphics[width=0.25\textwidth]{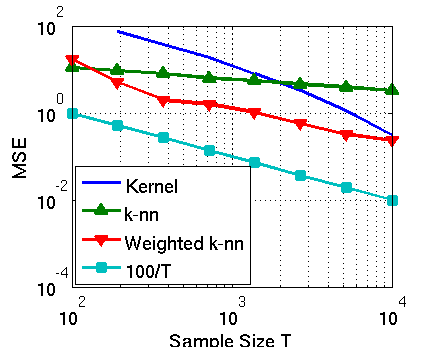}\includegraphics[width=0.25\textwidth]{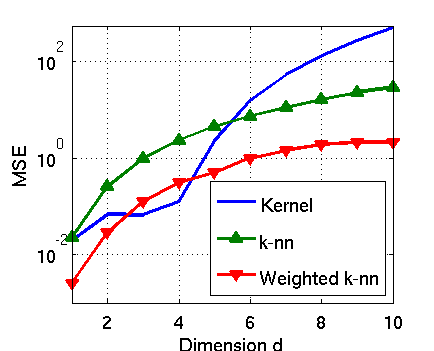}\caption{(Left) Log-log plot of MSE of the estimators for various sample sizes
with fixed $d=5$ for the non-truncated case and the theoretical bound
scaled by a constant $(100/T)$. (Right) Plot of MSE of the estimators
for various dimensions at a fixed sample size $T=3000$ for the non-truncated
case. 100 trials are used in both cases. The performance of the proposed
estimator is similar to that of the truncated case.\label{fig:notrunc_mse_n}}

\end{figure}

\section{Conclusion\label{sec:Conclusion}}

In this paper we derived convergence rates for a plug-in estimator
of $f$-divergence using $d$-dimensional truncated $k$-nn density
estimators. We then applied the theory of optimally weighted ensemble
estimation to obtain an estimator with a convergence rate of $O\left(\frac{1}{T}\right)$.
The advantages of this estimator is it is simple to implement, converges
rapidly, and performs well for higher dimensions. This weighted ensemble
divergence estimator also performs well for densities with unbounded
support.

\appendices

\section{Proof of Theorem~\ref{thm:bias}}

\label{sec:bias}Note that $\bias\left(\gk\right)=\mathbb{E}\left[g\left(\lhat(\mathbf{Z})\right)-g\left(\mathbb{E}_{\mathbf{Z}}\lhat(\mathbf{Z})\right)\right]+\mathbb{E}\left[g\left(\mathbb{E}_{\mathbf{Z}}\lhat(\mathbf{Z})\right)-g\left(L(\mathbf{Z})\right)\right].$
We find bounds for these terms by using Taylor series expansions.
The Taylor series expansion of $g\left(\lhat(\mathbf{Z})\right)$
around $\mathbb{E}_{\mathbf{Z}}\lhat(\mathbf{Z})$ gives\begin{equation}
g\left(\lhat(\mathbf{Z})\right)=\sum_{i=0}^{2}\frac{g^{(i)}\left(\mathbb{E}_{\mathbf{Z}}\lhat(\mathbf{Z})\right)}{i!}\ek^{i}(\mathbf{Z})+\frac{1}{6}g^{(3)}\left(\mathbf{\xi_{Z}}\right)\ek^{3}(\mathbf{Z})\label{eq:g_taylor}\end{equation}
where $\mathbf{\xi_{Z}}\in\left(\mathbb{E}_{\mathbf{Z}}\lhat(\mathbf{Z}),\lhat(\mathbf{Z})\right)$
comes from the mean value theorem. The following lemma enables us
to find bounds on the $\ek^{i}$ terms:
\begin{lem}
\label{lem:ekhat}Let $\gamma(z)$ be an arbitrary function with $\sup_{z}|\gamma(z)|<\infty.$
Let $\mathbf{Z}$ be a realization of the density $f_{2}$ independent
of $\fhat{i}$ for $i=1,2$. Then,\begin{eqnarray}
\mathbb{E}\left[\gamma(\mathbf{Z})\ehat i^{q}(\mathbf{Z})\right] & = & \begin{cases}
1_{\{q=2\}}\left(c_{2,i}(\gamma(z))\left(\frac{1}{k_{i}}\right)+o\left(\frac{1}{k_{i}}\right)\right)+1_{\{q\geq3\}}O\left(\frac{1}{k_{i}^{\frac{q}{2}}}\right), & q\geq2\\
0, & q=1,\end{cases}\label{eq:moment}\\
\mathbb{E}\left[\gamma(\mathbf{Z})\ehat 1^{q}(\mathbf{Z})\ehat 2^{r}(\mathbf{Z})\right] & = & \begin{cases}
O\left(\frac{1}{k_{1}^{\frac{q}{2}}k_{2}^{\frac{r}{2}}}\right), & q,\, r\geq2\\
0, & q=1\,\text{or }r=1\end{cases}\label{eq:cross_moment}\\
\mathbb{E}\left[\gamma(\mathbf{Z})\ek^{q}(\mathbf{Z})\right] & = & 1_{\{q=1\}}c_{4,1}\left(\frac{1}{k_{2}}\right)+1_{\{q=2\}}\left(c_{4,2}\left(\frac{1}{k_{2}}\right)+c_{4,3}\left(\frac{1}{k_{1}}\right)\right)+1_{\{q\geq3\}}O\left(\frac{1}{k_{1}^{\frac{q}{2}}}+\frac{1}{k_{2}^{\frac{q}{2}}}\right)\nonumber \\
 & = & 1_{\{q=1\}}c_{4,1}\left(\frac{1}{k_{2}}\right)+1_{\{q=2\}}\left(c_{4,2}\left(\frac{1}{k_{2}}\right)+c_{4,3}\left(\frac{1}{k_{1}}\right)\right)+1_{\{q\geq2\}}o\left(\frac{1}{k_{1}}\right)+o\left(\frac{1}{k_{2}}\right)\label{eq:ekhat}\end{eqnarray}

where $c_{2,i}$ and $c_{4,j}$ are functionals of $\gamma$, $f_{1},$
and $f_{2}.$ \end{lem}
\begin{proof}
For $i=2$, Eq.~\ref{eq:moment} is given and proved as Lemma 5 in~\cite{sricharan2013ensemble}
where the density estimator is a truncated uniform kernel density
estimator with bandwidth $\left(k/M\right)^{1/d}$. The proof uses
concentration inequalities to bound $\ez\ehat 2^{q}(\mathbf{Z})$
in terms of $k_{2}.$ It can then be shown that the $k$-nn density
estimator converges to a truncated uniform kernel density estimator~\cite{kumar2012thesis}.\textbf{\emph{
}}Thus the result holds for the $k$-nn density estimator as well.
For $i=1,$ the proof follows the same procedure but results in a
different constant.

For Eq.~\ref{eq:cross_moment}, note that for $q,\, r\geq2,$\begin{eqnarray*}
\mathbb{E}\left[\gamma(\mathbf{Z})\ehat 1^{q}(\mathbf{Z})\ehat 2^{r}(\mathbf{Z})\right] & = & \mathbb{E}\left[\gamma(\mathbf{Z})\mathbb{E}_{\mathbf{Z}}\left[\ehat 1^{q}(\mathbf{Z})\ehat 2^{r}(\mathbf{Z})\right]\right]\\
 & = & \mathbb{E}\left[\gamma(\mathbf{Z})\mathbb{E}_{\mathbf{Z}}\left[\ehat 1^{q}(\mathbf{Z})\right]\mathbb{E}_{\mathbf{Z}}\left[\ehat 2^{r}(\mathbf{Z})\right]\right]\\
 & = & \mathbb{E}\left[\gamma(\mathbf{Z})\left(O\left(\frac{1}{k_{1}^{\frac{q}{2}}k_{2}^{\frac{r}{2}}}\right)\right)\right]\\
 & = & O\left(\frac{1}{k_{1}^{\frac{q}{2}}k_{2}^{\frac{r}{2}}}\right),\end{eqnarray*}
 where we use conditional independence for the second equality and
Eq.~\ref{eq:moment} for the third equality. If either $q=0$ or
$r=0$ (but not both), then Eq.~\ref{eq:cross_moment} reduces to
Eq.~\ref{eq:moment}.

For Eq.~\ref{eq:ekhat}, we expand $\lhat(\mathbf{Z})$ around $\mathbb{E}_{\mathbf{Z}}\fhat{1}(\mathbf{Z})$
and $\mathbb{E}_{\mathbf{Z}}\fhat{2}(\mathbf{Z})$:\begin{eqnarray}
\frac{\fhat{1}(\mathbf{Z})}{\fhat{2}(\mathbf{Z})} & = & \frac{\mathbb{E}_{\mathbf{Z}}\fhat{1}(\mathbf{Z})}{\mathbb{E}_{\mathbf{Z}}\fhat{2}(\mathbf{Z})}+\frac{\ehat 1(\mathbf{Z})}{\mathbb{E}_{\mathbf{Z}}\fhat{2}(\mathbf{Z})}-\mathbb{E}_{\mathbf{Z}}\fhat{1}(\mathbf{Z})\frac{\ehat 2(\mathbf{Z})}{\left(\mathbb{E}_{\mathbf{Z}}\fhat{2}(\mathbf{Z})\right)^{2}}\nonumber \\
 &  & -\frac{\ehat 1(\mathbf{Z})\ehat 2(\mathbf{Z})}{\left(\mathbb{E}_{\mathbf{Z}}\fhat{2}(\mathbf{Z})\right)^{2}}+\mathbb{E}_{\mathbf{Z}}\fhat{1}(\mathbf{Z})\frac{\ehat 2^{2}(\mathbf{Z})}{2\left(\mathbb{E}_{\mathbf{Z}}\fhat{2}(\mathbf{Z})\right)^{3}}\nonumber \\
 &  & +\frac{\ehat 1(\mathbf{Z})\ehat 2^{2}(\mathbf{Z})}{2\left(\mathbb{E}_{\mathbf{Z}}\fhat{2}(\mathbf{Z})\right)^{3}}+o\left(\ehat 2^{2}(\mathbf{Z})+\ehat 1(\mathbf{Z})\ehat 2^{2}(\mathbf{Z})\right)\label{eq:lhat_taylor}\\
 & = & \frac{\mathbb{E}_{\mathbf{Z}}\fhat{1}(\mathbf{Z})}{\mathbb{E}_{\mathbf{Z}}\fhat{2}(\mathbf{Z})}+h(\ehat 1(\mathbf{Z}),\ehat 2(\mathbf{Z})).\nonumber \end{eqnarray}
 Let $\mathbf{h}(\mathbf{Z})=h(\ehat 1(\mathbf{Z}),\ehat 2(\mathbf{Z})).$
Thus $\ek(\mathbf{Z})=\frac{\mathbb{E}_{\mathbf{Z}}\fhat{1}(\mathbf{Z})}{\mathbb{E}_{\mathbf{Z}}\fhat{2}(\mathbf{Z})}-\mathbb{E}_{\mathbf{Z}}\lhat(\mathbf{Z})+\mathbf{h}(\mathbf{Z}).$
By the binomial theorem, \begin{equation}
\ek^{q}(\mathbf{Z})=\sum_{j=0}^{q}a_{q,j}\left(\frac{\mathbb{E}_{\mathbf{Z}}\fhat{1}(\mathbf{Z})}{\mathbb{E}_{\mathbf{Z}}\fhat{2}(\mathbf{Z})}-\mathbb{E}_{\mathbf{Z}}\lhat(\mathbf{Z})\right)^{q-j}\mathbf{h}^{j}(\mathbf{Z}),\label{eq:f_binomial}\end{equation}
where $a_{q,j}$ is the binomial coefficient. From~\cite{sricharan2013ensemble},
$\mathbb{E}_{\mathbf{Z}}\fhat{i}(\mathbf{Z})=f_{i}(\mathbf{Z})+\sum_{j=1}^{d}c_{i,j,k_{i}}(\mathbf{Z})\left(\frac{k_{i}}{M_{i}}\right)^{j/d}+o\left(\frac{k_{i}}{M_{i}}\right)=f_{i}(\mathbf{Z})+c_{1,i}(\mathbf{Z},k_{i},M_{i})=f_{i}(\mathbf{Z})+o(1).$
This quantity is bounded above and below based on our assumptions.
Using a Taylor series expansion of $\frac{1}{x}$ about $\mathbb{E}_{\mathbf{Z}}\fhat{2}(\mathbf{Z})$,
\begin{eqnarray}
\mathbb{E}_{\mathbf{Z}}\frac{1}{\fhat{2}(\mathbf{Z})} & = & \mathbb{E}_{\mathbf{Z}}\left[\frac{1}{\mathbb{E}_{\mathbf{Z}}\fhat{2}(\mathbf{Z})}-\frac{\ehat 2}{\left(\mathbb{E}_{\mathbf{Z}}\fhat{2}(\mathbf{Z})\right)^{2}}+\frac{\ehat 2^{2}}{2\xi_{2,\mathbf{Z}}}\right]\nonumber \\
 & = & \frac{1}{\mathbb{E}_{\mathbf{Z}}\fhat{2}(\mathbf{Z})}+\frac{\left(\var_{\mathbf{Z}}\left[\fhat{2}(\mathbf{Z})\right]\right)}{2\xi_{2,\mathbf{Z}}}\nonumber \\
 & = & \frac{1}{\mathbb{E}_{\mathbf{Z}}\fhat{2}(\mathbf{Z})}+c_{3,2}(\mathbf{Z})\left(\frac{1}{k_{2}}\right),\label{eq:kde}\end{eqnarray}
where $\xi_{2,\mathbf{Z}}\in\left(\mathbb{E}_{\mathbf{Z}}\fhat{2}(\mathbf{Z}),\fhat{2}(\mathbf{Z})\right)$
from the mean value thoerem and we use the fact that the variance
of the kernel density estimate converges to zero with rate $\frac{1}{M_{2}\sigma_{2}}$
where $\sigma_{2}=O\left(\frac{k_{2}}{M_{2}}\right)$. Thus \begin{eqnarray}
\left(\frac{\mathbb{E}_{\mathbf{Z}}\fhat{1}(\mathbf{Z})}{\mathbb{E}_{\mathbf{Z}}\fhat{2}(\mathbf{Z})}-\mathbb{E}_{\mathbf{Z}}\lhat(\mathbf{Z})\right)^{q} & = & \left(\mathbb{E}_{\mathbf{Z}}\fhat{1}(\mathbf{Z})c_{3,2}(\mathbf{Z})\left(\frac{1}{k_{2}}\right)\right)^{q}\nonumber \\
 & = & \left(f_{1}(\mathbf{Z})c_{3,2}(\mathbf{Z})\left(\frac{1}{k_{2}}\right)+\sum_{j=1}^{d}c_{1,j,k_{1}}\left(\frac{k_{1}}{M_{1}}\right)^{\frac{j}{d}}\left(\frac{1}{k_{2}}\right)+o\left(\frac{k_{1}}{M_{1}k_{2}}\right)\right)^{q}\nonumber \\
 & = & 1_{\{q=1\}}c_{3}(\mathbf{Z})\left(\frac{1}{k_{2}}\right)+1_{\{q\geq2\}}O\left(\frac{1}{k_{2}^{q}}\right)+o\left(\frac{1}{k_{2}^{q}}\right)=:\bz{q}(\mathbf{Z}).\label{eq:bias_final1}\end{eqnarray}
This is also bounded. Combining Eqs.~\ref{eq:f_binomial} and~\ref{eq:bias_final1}:
\begin{eqnarray}
\ek^{q}(\mathbf{Z}) & = & \bz{q}(\mathbf{Z})+\bz{q-1}^{1_{\{q\geq2\}}}(\mathbf{Z})a_{q,1}\mathbf{h}(\mathbf{Z})+1_{\{q\geq2\}}\bz{q-2}^{1_{\{q\geq3\}}}(\mathbf{Z})a_{q,2}\mathbf{h}^{2}(\mathbf{Z})+1_{\{q\geq3\}}\bz{q-3}^{1_{\{q\geq4\}}}(\mathbf{Z})O\left(\mathbf{h}^{3}(\mathbf{Z})\right)\nonumber \\
 & = & \bz{q}(\mathbf{Z})+\bz{q-1}^{1_{\{q\geq2\}}}(\mathbf{Z})a_{q,1}\times\nonumber \\
 &  & \left(\frac{\ehat 1(\mathbf{Z})}{\ez\fhat{2}(\mathbf{Z})}-\frac{\ez\fhat{1}(\mathbf{Z})}{\left(\ez\fhat{2}(\mathbf{Z})\right)^{2}}\ehat 2(\mathbf{Z})-\frac{\ehat 1(\mathbf{Z})\ehat 2(\mathbf{Z})}{\left(\ez\fhat{2}(\mathbf{Z})\right)^{2}}+\frac{\ez\fhat{1}(\mathbf{Z})}{2\left(\ez\fhat{2}(\mathbf{Z})\right)^{3}}\ehat 2^{2}(\mathbf{Z})+o\left(\ehat 2^{2}(\mathbf{Z})\right)\right)\nonumber \\
 &  & +1_{\{q\geq2\}}\bz{q-2}^{1_{\{q\geq3\}}}(\mathbf{Z})a_{q,2}\left(\frac{\ehat 1^{2}(\mathbf{Z})}{\left(\ez\fhat{2}(\mathbf{Z})\right)^{2}}+\frac{\left(\ez\fhat{1}(\mathbf{Z})\right)^{2}}{\left(\ez\fhat{2}(\mathbf{Z})\right)^{4}}\ehat 2^{2}(\mathbf{Z})+O\left(\ehat 1(\mathbf{Z})\ehat 2(\mathbf{Z})+\ehat 2^{3}(\mathbf{Z})\right)\right)\nonumber \\
 &  & +1_{\{q\geq3\}}\bz{q-2}^{1_{\{q\geq4\}}}(\mathbf{Z})\left(O\left(\ehat 1^{3}(\mathbf{Z})+\ehat 2^{3}(\mathbf{Z})+\ehat 1^{2}(\mathbf{Z})\ehat 2^{2}(\mathbf{Z})\right)\right)\nonumber \\
 & = & \bz{q}(\mathbf{Z})+\sum_{i=1}^{3}\mathbf{u}_{i,q}(\mathbf{Z}).\label{eq:Fk_q}\end{eqnarray}
Applying Eqs.~\ref{eq:moment} and~\ref{eq:cross_moment} we have\begin{eqnarray*}
\mathbb{E}\left[\gamma(\mathbf{Z})\ek^{q}(\mathbf{Z})\right] & = & \left(1_{\{q=1\}}\left(\mathbb{E}\left[\gamma(\mathbf{Z})c_{3}(\mathbf{Z})\right]+c_{2,2}\left(\frac{\gamma(z)\mathbb{E}_{Z}\fhat{1}(z)}{2\left(\mathbb{E}_{Z}\fhat{2}(z)\right)^{3}}\right)\right)+1_{\{q=2\}}c_{2,2}\left(\frac{\gamma(z)\left(\mathbb{E}_{Z}\fhat{1}(z)\right)^{2}}{\left(\mathbb{E}_{Z}\fhat{2}(z)\right)^{4}}\right)\right)\left(\frac{1}{k_{2}}\right)\\
 &  & +1_{\{q=2\}}c_{2,1}\left(\frac{\gamma(z)}{\left(\mathbb{E}_{Z}\fhat{2}(z)\right)^{2}}\right)\left(\frac{1}{k_{1}}\right)+1_{\{q\geq2\}}o\left(\frac{1}{k_{1}}\right)+o\left(\frac{1}{k_{2}}\right)\\
 & = & 1_{\{q=1\}}c_{4,1}\left(\frac{1}{k_{2}}\right)+1_{\{q=2\}}\left(c_{4,2}\left(\frac{1}{k_{2}}\right)+c_{4,3}\left(\frac{1}{k_{1}}\right)\right)+1_{\{q\geq2\}}o\left(\frac{1}{k_{1}}\right)+o\left(\frac{1}{k_{2}}\right).\end{eqnarray*}
Then since $\mathbb{E}_{Z}\fhat{i}(Z)=f_{i}(Z)+o(1)$, the constants
depend on $f_{1},$ $f_{2}$, and $\gamma.$ 

To obtain the more general bound for $\mathbb{E}\left[\gamma(\mathbf{Z})\ek^{q}(\mathbf{Z})\right]$,
note that from Eqs.~\ref{eq:moment}, \ref{eq:cross_moment}, and~\ref{eq:lhat_taylor},
the leading terms $\ez\mathbf{h}^{q}(\mathbf{Z})$ with $q<4$ are
\[
\ez\mathbf{h}^{q}(\mathbf{Z})=\begin{cases}
O\left(\frac{1}{k_{2}}\right), & q=1,\\
O\left(\frac{1}{k_{1}^{\frac{q}{2}}}+\frac{1}{k_{2}^{\frac{q}{2}}}+\frac{1}{k_{1}k_{2}}\right), & q=2,\,3.\end{cases}\]
Note that $O\left(\frac{1}{k_{1}k_{2}}\right)=O\left(\frac{1}{\left(\min\left(k_{1},k_{2}\right)\right)^{2}}\right)$.
Thus for $q=2,\,3,$ we ignore this term to get $\ez\mathbf{h}^{q}(\mathbf{Z})=O\left(\frac{1}{k_{1}^{q/2}}+\frac{1}{k_{2}^{q/2}}\right).$
For $q\geq4$, the leading terms come from products of powers of $\ehat 1$
and $\ehat 2.$ This gives \begin{eqnarray*}
\ez\mathbf{h}^{q}(\mathbf{Z}) & = & O\left(\frac{1}{k_{1}^{\frac{q}{2}}}+\frac{1}{k_{2}^{\frac{q}{2}}}+\sum_{\substack{i+j=q\\
i,j\geq2}
}\frac{1}{k_{1}^{\frac{i}{2}}k_{2}^{\frac{j}{2}}}\right)\\
 & = & O\left(\frac{1}{k_{1}^{\frac{q}{2}}}+\frac{1}{k_{2}^{\frac{q}{2}}}+\sum_{\substack{i+j=q\\
i,j\geq2}
}\frac{1}{\left(\min\left(k_{1},k_{2}\right)\right)^{\frac{i+j}{2}}}\right)\\
 & = & O\left(\frac{1}{k_{1}^{\frac{q}{2}}}+\frac{1}{k_{2}^{\frac{q}{2}}}\right)\\
\implies\mathbb{E}\left[\gamma(\mathbf{Z})\ek^{q}(\mathbf{Z})\right] & = & \sum_{j=0}^{q}O\left(\frac{1}{k_{2}^{q-j}}\right)\mathbb{E}\left[\ez\mathbf{h}^{j}(\mathbf{Z})\right]\\
 & = & 1_{\{q=1\}}O\left(\frac{1}{k_{2}}\right)+1_{\{q\geq2\}}O\left(\frac{1}{k_{1}^{\frac{q}{2}}}+\frac{1}{k_{2}^{\frac{q}{2}}}\right).\end{eqnarray*}
 
\end{proof}
The following lemma is required to bound the $g^{(3)}\left(\mathbf{\xi_{Z}}\right)$
term. 
\begin{lem}
\label{lem:bounded}Assume that U(x) is any arbitrary functional which
satisfies\begin{eqnarray*}
(i) & \mathbb{E}\left[\sup_{L\in(p_{l},p_{u})}\left|U\left(L\frac{\mathbf{p}_{2}}{\mathbf{p}_{1}}\right)\right|\right]=G_{1}<\infty,\\
(ii) & \sup_{L\in\left(\frac{q_{l,1}}{q_{u,2}},\frac{q_{u,1}}{q_{l,2}}\right)}\left|U\left(L\right)\right|\mathcal{C}\left(k_{1}\right)\mathcal{C}\left(k_{2}\right)=G_{2}<\infty,\\
(iii) & \mathbb{E}\left[\sup_{L\in\left(\frac{q_{l,1}}{p_{u,2}},\frac{q_{u,1}}{p_{l,2}}\right)}\left|U\left(L\mathbf{p}_{2}\right)\right|\mathcal{C}\left(k_{1}\right)\right]=G_{3}<\infty,\\
(iv) & \mathbb{E}\left[\sup_{L\in\left(\frac{p_{l,1}}{q_{u,2}},\frac{p_{u,1}}{q_{l,2}}\right)}\left|U\left(\frac{L}{\mathbf{p}_{1}}\right)\right|\mathcal{C}\left(k_{2}\right)\right]=G_{4}<\infty.\end{eqnarray*}
Let $\mathbf{Z}$ be $\mathbf{X}_{i}$ for some fixed $i\in\{1,\dots,N\}$
and $\xi_{\mathbf{Z}}$ be any random variable which almost surely
lies in $(L(\mathbf{Z}),\lhat(\mathbf{Z})).$ Then $\mathbb{E}|U(\xi_{\mathbf{Z}})|<\infty.$ \end{lem}
\begin{proof}
This is a version of Lemma~9 in~\cite{sricharan2013ensemble} modified
to apply to functionals of the likelihood ratio. Because of assumption
$\mathcal{A}.1$, it is sufficient to show that the conditional expectation
$\mathbb{E}\left[|U(\xi_{Z})|\,|\,\mathbf{X}_{1},\dots,\mathbf{X}_{N}\right]<\infty.$ 

First, some properties of $k$-NN density estimators are required.
Let $\mathbf{S}_{k_{i},i}(Z)=\left\{ Y:d(Z,Y)\leq\mathbf{d}_{Z,i}^{(k_{i})}\right\} $
where $\mathbf{d}_{Z,i}^{(k_{i})}$ is the distance to the $k_{i}$th
nearest neighbor of $Z$ from the corresponding set of samples. Then
let $\mathbf{P}_{i}(Z)=\int_{\mathbf{S}_{k_{i},i}(Z)}f_{i}(x)dx$
which has a beta distribution with parameters $k_{i}$ and $M_{i}-k_{i}+1$~\cite{mack1979knn}.
Let $A_{i}(Z)$ be the event that $\mathbf{P}_{i}(Z)<\left(\frac{\sqrt{6}}{k_{i}^{\delta/2}}+1\right)\frac{k_{i}-1}{M_{i}}.$
It has been shown that $Pr\left(A_{i}(Z)^{C}\right)=\Theta\left(\mathcal{C}\left(k_{i}\right)\right)$
and that under $A_{i}(Z)$~\cite{kumar2012thesis,sricharan2012estimation},
\[
\frac{p_{l,i}}{\mathbf{P}_{i}(Z)}<\fhat{i}(Z)<\frac{p_{u,i}}{\mathbf{P}_{i}(Z)}.\]
It has also been shown that under $A_{i}(Z)^{C}$~\cite{kumar2012thesis,sricharan2012estimation},
\[
q_{l,i}<\fhat{i}(Z)<q_{u,i}.\]
Let $A(Z)=A_{1}(Z)\cap A_{2}(Z)$ and note that $A_{1}(Z)$ and $A_{2}(Z)$
are independent events. Thus since $\lhat(Z)=\frac{\fhat{1}(Z)}{\fhat{2}(Z)},$
we have that under $A(Z),$ \[
p_{l}\frac{\mathbf{P}_{2}(Z)}{\mathbf{P}_{1}(Z)}<\lhat(Z)<p_{u}\frac{\mathbf{P}_{2}(Z)}{\mathbf{P}_{1}(Z)}.\]
Now let $Q_{1}(Z)=A_{1}(Z)^{C}\cap A_{2}(Z)^{C},$ $Q_{2}(Z)=A_{1}(Z)^{C}\cap A_{2}(Z),$
and $Q_{3}(Z)=A_{1}(Z)\cap A_{2}(Z)^{C}.$ Then due to independence
and the fact that the $Q_{i}(Z)$s are disjoint, \begin{eqnarray*}
A(Z)^{C} & = & A_{1}(Z)^{C}\cup A_{2}(Z)^{C}=Q_{1}(Z)\cup Q_{2}(Z)\cup Q_{3}(Z),\\
\implies Pr\left(A(Z)^{C}\right) & = & Pr\left(Q_{1}(Z)\right)+Pr\left(Q_{2}(Z)\right)+Pr\left(Q_{3}(Z)\right)\\
 & \leq & \mathcal{C}\left(k_{1}\right)\mathcal{C}\left(k_{2}\right)+\mathcal{C}\left(k_{1}\right)+\mathcal{C}\left(k_{2}\right).\end{eqnarray*}
Then under $Q_{1}(Z)$, $Q_{2}(Z),$ and $Q_{3}(Z)$, respectively,\[
\frac{q_{l,1}}{q_{u,2}}<\lhat(Z)<\frac{q_{u,1}}{q_{l,2}},\]
\[
\frac{q_{l,1}\mathbf{P}_{2}(Z)}{p_{u,2}}<\lhat(Z)<\frac{q_{u,1}\mathbf{P}_{2}(Z)}{p_{l,2}},\]
\[
\frac{p_{l,1}}{\mathbf{P}_{1}(Z)q_{u,2}}<\lhat(Z)<\frac{p_{u,1}}{\mathbf{P}_{1}(Z)q_{l,2}}.\]
Conditioning on $\mathbf{X}_{1},\dots,\mathbf{X}_{N}$ gives \begin{eqnarray*}
\mathbb{E}\left[\left|U(\xi_{Z})\right|\right] & = & \mathbb{E}\left[1_{A(Z)}\left|U(\xi_{Z})\right|\right]+\mathbb{E}\left[1_{Q_{1}(Z)}\left|U(\xi_{Z})\right|\right]+\mathbb{E}\left[1_{Q_{2}(Z)}\left|U(\xi_{Z})\right|\right]+\mathbb{E}\left[1_{Q_{3}(Z)}\left|U(\xi_{Z})\right|\right]\\
 & \leq & Pr(A(Z))\mathbb{E}\left[\sup_{L\in(p_{l},p_{u})}\left|U\left(L\frac{\mathbf{P}_{2}(Z)}{\mathbf{P}_{1}(Z)}\right)\right|\right]+Pr(Q_{1}(Z))\sup_{L\in\left(\frac{q_{l,1}}{q_{u,2}},\frac{q_{u,1}}{q_{l,2}}\right)}\left|U\left(L\right)\right|\\
 &  & +Pr(Q_{1}(Z))\mathbb{E}\left[\sup_{L\in\left(\frac{q_{l,1}}{p_{u,2}},\frac{q_{u,1}}{p_{l,2}}\right)}\left|U\left(L\mathbf{P}_{2}(Z)\right)\right|\right]+Pr(Q_{1}(Z))\mathbb{E}\left[\sup_{L\in\left(\frac{p_{l,1}}{q_{u,2}},\frac{p_{u,1}}{q_{l,2}}\right)}\left|U\left(\frac{L}{\mathbf{P}_{1}(Z)}\right)\right|\right]\\
 & \leq & \mathbb{E}\left[\sup_{L\in(p_{l},p_{u})}\left|U\left(L\frac{\mathbf{P}_{2}(Z)}{\mathbf{P}_{1}(Z)}\right)\right|\right]+\sup_{L\in\left(\frac{q_{l,1}}{q_{u,2}},\frac{q_{u,1}}{q_{l,2}}\right)}\left|U\left(L\right)\right|\mathcal{C}\left(k_{1}\right)\mathcal{C}\left(k_{2}\right)\\
 &  & +\mathbb{E}\left[\sup_{L\in\left(\frac{q_{l,1}}{p_{u,2}},\frac{q_{u,1}}{p_{l,2}}\right)}\left|U\left(L\mathbf{P}_{2}(Z)\right)\right|\mathcal{C}\left(k_{1}\right)\right]+\mathbb{E}\left[\sup_{L\in\left(\frac{p_{l,1}}{q_{u,2}},\frac{p_{u,1}}{q_{l,2}}\right)}\left|U\left(\frac{L}{\mathbf{P}_{1}(Z)}\right)\right|\mathcal{C}\left(k_{2}\right)\right]\\
 & = & G_{1}+G_{2}+G_{3}+G_{4}<\infty.\end{eqnarray*}

\end{proof}
Applying Lemma~\ref{lem:bounded} and assumption $(\mathcal{A}.5)$
gives $\mathbb{E}\left[\left(g^{(3)}\left(\xi_{\mathbf{Z}}\right)/6\right)^{2}\right]=O(1).$
Then by Cauchy-Schwarz and applying Lemma~\ref{lem:ekhat},\[
\mathbb{E}\left[\frac{1}{6}g^{(3)}\left(\mathbf{\xi_{Z}}\right)\ek^{3}(\mathbf{Z})\right]\leq\sqrt{\mathbb{E}\left[\left(\frac{g^{(3)}\left(\xi_{\mathbf{Z}}\right)}{6}\right)^{2}\right]\mathbb{E}\left[\ek^{6}(\mathbf{Z})\right]}=o\left(\frac{1}{k_{1}}+\frac{1}{k_{2}}\right).\]
Using this result with Eq.~\ref{eq:g_taylor} and applying Lemma~\ref{lem:ekhat}
again gives \begin{equation}
\mathbb{E}\left[g\left(\lhat(\mathbf{Z})\right)-g\left(\ez\lhat(\mathbf{Z})\right)\right]=\left(c_{4,1}+c_{4,2}\right)\left(\frac{1}{k_{2}}\right)+c_{4,3}\left(\frac{1}{k_{1}}\right)+o\left(\frac{1}{k_{1}}+\frac{1}{k_{2}}\right).\label{eq:bias1}\end{equation}

Now by Taylor series expansion \[
g\left(\ez\lhat(\mathbf{Z})\right)=g\left(L(\mathbf{Z})\right)+\sum_{i=1}^{d}g^{(i)}\left(L(\mathbf{Z})\right)\left(\ez\lhat(\mathbf{Z})-L(\mathbf{Z})\right)^{i}+o\left(\left(\ez\lhat(\mathbf{Z})-L(\mathbf{Z})\right)^{d}\right).\]
From Eq.~\ref{eq:kde}, \begin{eqnarray*}
\ez\lhat(\mathbf{Z})-L(\mathbf{Z}) & = & \ez\fhat{1}(\mathbf{Z})\left(\frac{1}{f_{2}(\mathbf{Z})+c_{1,2}(\mathbf{Z},k_{2},M_{2})}+c_{3,2}(\mathbf{Z})\left(\frac{1}{k_{2}}\right)\right)-L(\mathbf{Z})\\
 & = & \frac{f_{2}(\mathbf{Z})c_{1,1}(\mathbf{Z},k_{1},M_{1})-f_{1}(\mathbf{Z})c_{1,2}(\mathbf{Z},k_{1},M_{2})}{f_{2}(\mathbf{Z})\left(f_{2}(\mathbf{Z})+c_{1,2}(\mathbf{Z},k_{2},M_{2})\right)}+\ez\fhat{1}(\mathbf{Z})c_{3,2}(\mathbf{Z})\left(\frac{1}{k_{2}}\right)\\
 & = & \frac{c_{1,1}(\mathbf{Z},k_{1},M_{1})}{f_{2}(\mathbf{Z})+o(1)}-\frac{f_{1}(\mathbf{Z})c_{1,2}(\mathbf{Z},k_{2},M_{2})}{f_{2}(\mathbf{Z})\left(f_{2}(\mathbf{Z})+o(1)\right)}+\ez\fhat{1}(\mathbf{Z})c_{3,2}(\mathbf{Z})\left(\frac{1}{k_{2}}\right)\\
 & = & \sum_{j=1}^{d}\left(c_{5,j,1}(\mathbf{Z})\left(\frac{k_{1}}{M_{1}}\right)^{\frac{j}{d}}+c_{5,j,2}(\mathbf{Z})\left(\frac{k_{2}}{M_{2}}\right)^{\frac{j}{d}}\right)+f_{1}(\mathbf{Z})c_{3,2}(\mathbf{Z})\left(\frac{1}{k_{2}}\right)+o\left(\frac{k_{1}}{M_{1}}+\frac{k_{2}}{M_{2}}+\frac{1}{k_{2}}\right).\end{eqnarray*}
This gives \begin{equation}
\mathbb{E}\left[g\left(\mathbb{E}_{\mathbf{Z}}\hat{\mathbf{L}}_{k}(\mathbf{Z})\right)-g\left(L(\mathbf{Z})\right)\right]=\sum_{j=1}^{d}\left(c_{6,j,1}\left(\frac{k_{1}}{M_{1}}\right)^{\frac{j}{d}}+c_{6,j,2}\left(\frac{k_{2}}{M_{2}}\right)^{\frac{j}{d}}\right)+c_{6,3}\left(\frac{1}{k_{2}}\right)+o\left(\frac{k_{1}}{M_{1}}+\frac{k_{2}}{M_{2}}+\frac{1}{k_{2}}\right),\label{eq:bias_final2}\end{equation}
where $c_{6,3}=\mathbb{E}\left[g'\left(L(\mathbf{Z})\right)f_{1}(\mathbf{Z})c_{3,2}(\mathbf{Z})\right]$
and $c_{6,j,i}$ is a functional of $g$, the derivatives of $g$,
and the densities $f_{1}$ and $f_{2}$.

Combining Eqs.~\ref{eq:bias1} and~\ref{eq:bias_final2} completes
the proof.

\section{Proof of Theorem~\ref{thm:variance}}

\label{sec:variance}Again, we start by forming a Taylor series expansion
of $g\left(\lhat(\mathbf{Z})\right)$ around $\mathbb{E}_{\mathbf{Z}}\lhat(\mathbf{Z})$.
\[
g\left(\lhat(\mathbf{Z})\right)=\sum_{i=0}^{\lambda-1}\frac{g^{(i)}\left(\mathbb{E}_{\mathbf{Z}}\lhat(\mathbf{Z})\right)}{i!}\ek^{i}(\mathbf{Z})+\frac{g^{(\lambda)}\left(\mathbf{\xi_{Z}}\right)}{\lambda!}\ek^{\lambda}(\mathbf{Z}),\]
 where $\mathbf{\xi_{Z}}\in\left(\mathbb{E}_{\mathbf{Z}}\lhat(\mathbf{Z}),\lhat(\mathbf{Z})\right)$.
Let $\Psi(\mathbf{Z})=g^{(\lambda)}\left(\mathbf{\xi_{Z}}\right)/\lambda!$
and define the operator $\mathcal{M}(\mathbf{Z})=\mathbf{Z}-\mathbb{E}\mathbf{Z}.$
Let \begin{eqnarray*}
\mathbf{p}_{i} & = & \mathcal{M}\left(g\left(\mathbb{E}_{\mathbf{X}_{i}}\lhat\left(\mathbf{X}_{i}\right)\right)\right),\\
\mathbf{q}_{i} & = & \mathcal{M}\left(g'\left(\mathbb{E}_{\mathbf{X}_{i}}\lhat\left(\mathbf{X}_{i}\right)\right)\ek\left(\mathbf{X}_{i}\right)\right),\\
\mathbf{r}_{i} & = & \mathcal{M}\left(\sum_{j=2}^{\lambda-1}\frac{g^{(j)}\left(\mathbb{E}_{\mathbf{X}_{i}}\lhat\left(\mathbf{X}_{i}\right)\right)}{j!}\ek^{j}\left(\mathbf{X}_{i}\right)\right)\\
\mathbf{s}_{i} & = & \mathcal{M}\left(\Psi\left(\mathbf{X}_{i}\right)\ek^{\lambda}\left(\mathbf{X}_{i}\right)\right).\end{eqnarray*}
 Then the variance of $\gk$ is \begin{eqnarray*}
\var\left(\gk\right) & = & \mathbb{E}\left[\left(\gk-\mathbb{E}\gk\right)^{2}\right]\\
 & = & \frac{1}{N}\mathbb{E}\left[\left(\mathbf{p}_{1}+\mathbf{q}_{1}+\mathbf{r}_{1}+\mathbf{s}_{1}\right)^{2}\right]+\frac{N-1}{N}\mathbb{E}\left[\left(\mathbf{p}_{1}+\mathbf{q}_{1}+\mathbf{r}_{1}+\mathbf{s}_{1}\right)\left(\mathbf{p}_{2}+\mathbf{q}_{2}+\mathbf{r}_{2}+\mathbf{s}_{2}\right)\right].\end{eqnarray*}
We will bound this using Lemma~\ref{lem:ekhat} and the following
lemmas. 
\begin{lem}
\label{lem:cov_split}Let $\Psi_{i}=\left\{ \{X,Y\}:\left\Vert X-Y\right\Vert _{1}\geq2\left(\frac{k_{i}}{M_{i}}\right)^{\frac{1}{d}}\right\} .$
For a fixed pair of points $\{X,Y\}\in\Psi_{i}$, and positive integers
$q,$ $r,$ \[
Cov\left[\ehat i^{q}(X),\ehat i^{r}(Y)\right]=1_{\{q=r=1\}}\left(\frac{-f_{i}(X)f_{i}(Y)}{M_{i}}\right)+o\left(\frac{1}{M_{i}}\right).\]
 For a fixed pair of points $\{X,Y\}\in\Psi_{i}^{C},$\[
Cov\left[\ehat i^{q}(X),\ehat i^{r}(Y)\right]=1_{\{q=r=1\}}O\left(\frac{1}{k_{i}}\right)+o\left(\frac{1}{k_{i}}\right).\]

\end{lem}
This lemma is given and proved as Lemmas~6 and 7 in~\cite{sricharan2013ensemble}
for the truncated uniform kernel density estimator using concentration
inequalities and Eq.~\ref{eq:moment}. Thus the result holds for
the $k$-nn density estimator as well.
\begin{lem}
\label{lem:cov_ekhat}Let $\gamma_{1}(x),$ $\gamma_{2}(x)$ be arbitrary
functions with $1$ partial derivative wrt $x$ and $\sup_{x}|\gamma_{i}(x)|<\infty,\, i=1,\,2.$
Let $\mathbf{X},$ $\mathbf{Y}$ be realizations of the density $f_{2}$
independent of the realizations used for $\fhat{1}$ and $\fhat{2}$.
Let $E_{0}=\{s,q,t,r\geq1\}$, $E_{1,1}=\{s=0,q\geq2,t\geq1,r\geq1\}\cup\{s\geq1,q\geq1,t=0,r\geq2\}$,
and $E_{1,2}=\{s\geq2,q=0,t\geq1,r\geq1\}\cup\{s\geq1,q\geq1,t\geq2,r=0\}$.
Then\begin{eqnarray}
Cov\left[\gamma_{1}(\mathbf{X})\ehat i^{q}(\mathbf{X}),\gamma_{2}(\mathbf{Y})\ehat i^{r}(\mathbf{Y})\right] & = & 1_{\{q=r=1\}}c_{7,i}(\gamma_{1}(x),\gamma_{2}(x))\left(\frac{1}{M_{i}}\right)+o\left(\frac{1}{M_{i}}\right),\label{eq:cov1}\\
Cov\left[\gamma_{1}(\mathbf{X})\ehat 1^{s}(\mathbf{X})\ehat 2^{q}(\mathbf{X}),\gamma_{2}(\mathbf{Y})\ehat 1^{t}(\mathbf{Y})\ehat 2^{r}(\mathbf{Y})\right] & = & \begin{cases}
o\left(\frac{1}{M_{1}}+\frac{1}{M_{2}}+\frac{1}{k_{1}^{2}}+\frac{1}{k_{2}^{2}}\right), & E_{0}\\
o\left(\frac{1}{\max\left(M_{1},M_{2}\right)}+\frac{1}{M_{2}}+\frac{1}{k_{1}^{2}}+\frac{1}{k_{2}^{2}}\right), & E_{1,1}\\
o\left(\frac{1}{\max\left(M_{1},M_{2}\right)}+\frac{1}{M_{1}}+\frac{1}{k_{1}^{2}}+\frac{1}{k_{2}^{2}}\right), & E_{1,2}\\
0, & \text{otherwise}\end{cases},\label{eq:cov2}\\
Cov\left[\gamma_{1}(\mathbf{X})\ek^{q}(\mathbf{X}),\gamma_{2}(\mathbf{Y})\ek^{r}(\mathbf{Y})\right] & = & 1_{\{q=1,r=1\}}\left(c_{8,1}\left(\gamma_{1}(x),\gamma_{2}(x)\right)\left(\frac{1}{M_{1}}\right)+c_{8,2}\left(\gamma_{1}(x),\gamma_{2}(x)\right)\left(\frac{1}{M_{2}}\right)\right)\nonumber \\
 &  & +o\left(\frac{1}{M_{1}}+\frac{1}{M_{2}}+\frac{1}{k_{1}^{2}}+\frac{1}{k_{2}^{2}}\right).\label{eq:cov3}\end{eqnarray}
 \end{lem}
\begin{proof}
Eq.~\ref{eq:cov1} is given and proved as Lemma 8 in~\cite{sricharan2013ensemble}
using results given in Lemma~\ref{lem:cov_split}. For Eq.~\ref{eq:cov2},
we have by Eqs.~\ref{eq:moment} and~\ref{eq:cross_moment} and
conditional independence when $E_{0},$ $E_{1,1}$, or $E_{1,2}$
hold:\begin{align*}
Cov\left[\gamma_{1}(X)\ehat 1^{s}(X)\ehat 2^{q}(X),\gamma_{2}(Y)\ehat 1^{t}(Y)\ehat 2^{r}(Y)\right]\end{align*}

\begin{eqnarray}
 & = & \mathbb{E}\left[\gamma_{1}(X)\ehat 1^{s}(X)\ehat 2^{q}(X)\gamma_{2}(Y)\ehat 1^{t}(Y)\ehat 2^{r}(Y)\right]\nonumber \\
 &  & -\mathbb{E}\left[\gamma_{1}(X)\ehat 1^{s}(X)\ehat 2^{q}(X)\right]\mathbb{E}\left[\gamma_{2}(Y)\ehat 1^{t}(Y)\ehat 2^{r}(Y)\right]\nonumber \\
 & = & \gamma_{1}(X)\gamma_{2}(Y)\mathbb{E}\left[\ehat 1^{s}(X)\ehat 1^{t}(Y)\right]\mathbb{E}\left[\ehat 2^{q}(X)\ehat 2^{r}(Y)\right]+1_{E_{0}\cap\{q,r,s,t\geq2\}}o\left(\frac{1}{\min\left(k_{1},k_{2}\right)^{2}}\right)\nonumber \\
 &  & +1_{E_{1,1}\cap\{\{t,r\geq2,s=0\}\cup\{s,q\geq2,t=0\}\}}o\left(\frac{1}{k_{2}^{2}}\right)+1_{E_{1,2}\cap\{\{t,r\geq2,q=0\}\cup\{s,q\geq2,r=0\}\}}o\left(\frac{1}{k_{1}^{2}}\right).\label{eq:cov_inside1}\end{eqnarray}
Note that $\mathbb{E}\left[\ehat i^{s}(X)\ehat i^{t}(Y)\right]=Cov\left[\ehat i^{s}(X),\ehat i^{t}(Y)\right]+\mathbb{E}\left[\ehat i^{s}(X)\right]\mathbb{E}\left[\ehat i^{t}(Y)\right].$
Consider the case where $E_{0}$ holds. By Eq.~\ref{eq:moment} and
Lemma~\ref{lem:cov_split}, this gives \begin{equation}
\mathbb{E}\left[\ehat i^{s}(X)\ehat i^{t}(Y)\right]=1_{\{s=t=2\}}O\left(\frac{1}{k_{i}^{2}}\right)+o\left(\frac{1}{k_{i}^{2}}\right)+\begin{cases}
1_{\left\{ s=t=1\right\} }\left(\frac{-f_{i}(X)f_{i}(Y)}{M_{i}}\right)+o\left(\frac{1}{M_{i}}\right), & \{X,Y\}\in\Psi_{i}\\
1_{\{s=t=1\}}O\left(\frac{1}{k_{i}}\right)+o\left(\frac{1}{k_{i}}\right), & \{X,Y\}\in\Psi_{i}^{c}.\end{cases}\label{eq:cov_inside2}\end{equation}
Note that \[
\mathbb{E}\left[Cov_{\mathbf{X,Y}}\left[\gamma_{1}(\mathbf{X})\ehat 1^{s}(\mathbf{X})\ehat 2^{q}(\mathbf{X}),\gamma_{2}(\mathbf{Y})\ehat 1^{t}(\mathbf{Y})\ehat 2^{r}(\mathbf{Y})\right]\right]=I_{1}+I_{2}+I_{3}+I_{4},\]
where \[
I_{1}=\mathbb{E}\left[1_{\{\mathbf{X},\mathbf{Y}\}\in\Psi_{1}^{C}\cap\Psi_{2}^{C}}\gamma_{1}(\mathbf{X})\gamma_{2}(\mathbf{Y})Cov_{\mathbf{X,Y}}\left[\ehat 1^{s}(\mathbf{X})\ehat 2^{q}(\mathbf{X}),\ehat 1^{t}(Y)\ehat 2^{r}(\mathbf{Y})\right]\right],\]
\[
I_{2}=\mathbb{E}\left[1_{\{\mathbf{X},\mathbf{Y}\}\in\Psi_{1}^{C}\cap\Psi_{2}}\gamma_{1}(\mathbf{X})\gamma_{2}(\mathbf{Y})Cov_{\mathbf{X,Y}}\left[\ehat 1^{s}(\mathbf{X})\ehat 2^{q}(\mathbf{X}),\ehat 1^{t}(\mathbf{Y})\ehat 2^{r}(\mathbf{Y})\right]\right],\]
\[
I_{3}=\mathbb{E}\left[1_{\{\mathbf{X},\mathbf{Y}\}\in\Psi_{1}\cap\Psi_{2}^{C}}\gamma_{1}(\mathbf{X})\gamma_{2}(\mathbf{Y})Cov_{\mathbf{X,Y}}\left[\ehat 1^{s}(\mathbf{X})\ehat 2^{q}(\mathbf{X}),\ehat 1^{t}(\mathbf{Y})\ehat 2^{r}(\mathbf{Y})\right]\right],\]
\[
I_{4}=\mathbb{E}\left[1_{\{\mathbf{X},\mathbf{Y}\}\in\Psi_{1}\cap\Psi_{2}}\gamma_{1}(\mathbf{X})\gamma_{2}(\mathbf{Y})Cov_{\mathbf{X,Y}}\left[\ehat 1^{s}(\mathbf{X})\ehat 2^{q}(\mathbf{X}),\ehat 1^{t}(\mathbf{Y})\ehat 2^{r}(\mathbf{Y})\right]\right].\]
Combining Eqs.~\ref{eq:cov_inside1} and~\ref{eq:cov_inside2} gives
\begin{eqnarray*}
I_{1} & = & \mathbb{E}\left[1_{\{\mathbf{X},\mathbf{Y}\}\in\Psi_{1}^{C}\cap\Psi_{2}^{C}}\gamma_{1}(\mathbf{X})\gamma_{2}(\mathbf{Y})\left(1_{\{q=r=s=t=1\}}O\left(\frac{1}{k_{1}k_{2}}\right)+o\left(\frac{1}{k_{1}k_{2}}\right)\right)\right]+o\left(\frac{1}{k_{1}^{2}}+\frac{1}{k_{2}^{2}}\right)\\
 & = & \int\left[\left(1_{\{q=r=s=t=1\}}O\left(\frac{1}{k_{1}k_{2}}\right)+o\left(\frac{1}{k_{1}k_{2}}\right)\right)\left(\gamma_{1}(x)\gamma_{2}(x)+o(1)\right)\right]\left(\int_{\{x,y\}\in\Psi_{1}^{C}\cap\Psi_{2}^{C}}dy\right)dx+o\left(\frac{1}{k_{1}^{2}}+\frac{1}{k_{2}^{2}}\right)\\
 & \leq & \int\left[\left(1_{\{q=r=s=t=1\}}O\left(\frac{1}{k_{1}k_{2}}\right)+o\left(\frac{1}{k_{1}k_{2}}\right)\right)\left(\gamma_{1}(x)\gamma_{2}(x)+o(1)\right)\right]\left(2^{d}\min_{i\in\{1,2\}}\frac{k_{i}}{M_{i}}\right)dx+o\left(\frac{1}{k_{1}^{2}}+\frac{1}{k_{2}^{2}}\right)\\
 & = & o\left(\frac{1}{\max\left(M_{1},M_{2}\right)}+\frac{1}{k_{1}^{2}}+\frac{1}{k_{2}^{2}}\right),\end{eqnarray*}
where $\arg\min_{i\in\{1,2\}}\frac{k_{i}}{M_{i}}=\arg\max\left(M_{1},M_{2}\right)$
because $k_{i}=k_{0}M_{i}^{\beta}$ by assumption $\left(\mathcal{A}.0\right)$.
Now also by Eqs. ~\ref{eq:cov_inside1} and~\ref{eq:cov_inside2},\[
I_{2}=\mathbb{E}\left[1_{\{\mathbf{X},\mathbf{Y}\}\in\Psi_{1}^{C}\cap\Psi_{2}}\gamma_{1}(\mathbf{X})\gamma_{2}(\mathbf{Y})o\left(\frac{1}{M_{2}}\right)\right]+o\left(\frac{1}{k_{1}^{2}}+\frac{1}{k_{2}^{2}}\right)=o\left(\frac{1}{M_{2}}+\frac{1}{k_{1}^{2}}+\frac{1}{k_{2}^{2}}\right).\]
Similarly, $I_{3}=o\left(\frac{1}{M_{1}}+\frac{1}{k_{1}^{2}}+\frac{1}{k_{2}^{2}}\right)$
and $I_{4}=o\left(\frac{1}{\max(M_{1},M_{2})}+\frac{1}{k_{1}^{2}}+\frac{1}{k_{2}^{2}}\right).$
Combining these results completes the proof for the case of $E_{0}$.

Now consider the case where $E_{1,1}$ holds. Specifically, assume
WLOG that $s=0$. Then Eq.~\ref{eq:cov_inside2} for $i=1$ gives
\begin{eqnarray*}
\mathbb{E}\left[\ehat 1^{t}(Y)\right] & = & 1_{\{t=2\}}O\left(\frac{1}{k_{1}}\right)+o\left(\frac{1}{k_{1}}\right),\\
\implies I_{1} & = & \mathbb{E}\left[1_{\{\mathbf{X},\mathbf{Y}\}\in\Psi_{1}^{C}\cap\Psi_{2}^{C}}\gamma_{1}(\mathbf{X})\gamma_{2}(\mathbf{Y})\left(1_{\{q=r=1,t=2\}}O\left(\frac{1}{k_{1}k_{2}}\right)+o\left(\frac{1}{k_{1}k_{2}}\right)\right)\right]+o\left(\frac{1}{k_{1}^{2}}+\frac{1}{k_{2}^{2}}\right)\\
 & = & o\left(\frac{1}{\max\left(M_{1},M_{2}\right)}+\frac{1}{k_{1}^{2}}+\frac{1}{k_{2}^{2}}\right).\end{eqnarray*}
 Similarly, since $\Psi_{1}\cap\Psi_{2}^{c}\subseteq\Psi_{2}^{c}$,
$I_{2},\, I_{3}=o\left(\frac{1}{M_{2}}+\frac{1}{k_{1}^{2}}+\frac{1}{k_{2}^{2}}\right),$
and $I_{4}=o\left(\frac{1}{M_{2}}+\frac{1}{k_{1}^{2}}+\frac{1}{k_{2}^{2}}\right).$
A similar argument for when $E_{1,2}$ holds shows that $I_{1}$ is
the same and that $I_{2},\, I_{3},$ and $I_{4}=o\left(\frac{1}{M_{1}}+\frac{1}{k_{1}^{2}}+\frac{1}{k_{2}^{2}}\right).$

Let $E_{2}=\{s,q=0;t,r\geq2\}\cup\{t,r=0;s,q\geq2\},$ $E_{3}=\{s,q,t=0;r\geq2\}\cup\{s,t,r=0;q\geq2\},$
and $E_{4}=\{q,t,r=0;s\geq2\}\cup\{s,q,r=0;t\geq2\}.$ Suppose that
$E_{2}$ holds and that WLOG $s,q=0$ and $t,r\geq2$. Then we have
\begin{eqnarray*}
Cov\left[\gamma_{1}(\mathbf{X}),\gamma_{2}(\mathbf{Y})\ehat 1^{t}(\mathbf{Y})\ehat 2^{r}(\mathbf{Y})\right] & = & \mathbb{E}\left[\gamma_{1}(\mathbf{X})\gamma_{2}(\mathbf{Y})\ehat 1^{t}(\mathbf{Y})\ehat 2^{r}(\mathbf{Y})\right]-\mathbb{E}\left[\gamma_{1}(\mathbf{X})\right]\mathbb{E}\left[\gamma_{2}(\mathbf{Y})\ehat 1^{t}(\mathbf{Y})\ehat 2^{r}(\mathbf{Y})\right]\\
 & = & \mathbb{E}\left[\gamma_{1}(\mathbf{X})\right]\mathbb{E}\left[\gamma_{2}(\mathbf{Y})\ehat 1^{t}(\mathbf{Y})\ehat 2^{r}(\mathbf{Y})\right]-\mathbb{E}\left[\gamma_{1}(\mathbf{X})\right]\mathbb{E}\left[\gamma_{2}(\mathbf{Y})\ehat 1^{t}(\mathbf{Y})\ehat 2^{r}(\mathbf{Y})\right]\\
 & = & 0,\end{eqnarray*}
 where we used the fact that $\mathbf{X}$ and $\mathbf{Y}$ are independent
to obtain the second inequality. The same result follows when either
$E_{3}$ or $E_{4}$ hold.

Proof of Eq.~\ref{eq:cov3}: from Eq.~\ref{eq:Fk_q}, \begin{eqnarray*}
Cov\left[\gamma_{1}(\mathbf{X})\ek^{q}(\mathbf{X}),\gamma_{2}(\mathbf{Y})\ek^{r}(\mathbf{Y})\right] & = & Cov\left[B_{1,q}(\mathbf{X}),B_{2,r}(\mathbf{Y})\right]+Cov\left[\mathbf{A}_{1,q}(\mathbf{X}),B_{2,r}(\mathbf{Y})\right]\\
 &  & +Cov\left[B_{1,q}(\mathbf{X}),\mathbf{A}_{2,r}(\mathbf{Y})\right]+Cov\left[\mathbf{A}_{1,q}(\mathbf{X}),\mathbf{A}_{2,r}(\mathbf{Y})\right],\end{eqnarray*}
where $B_{i,q}(\mathbf{X})=\gamma_{i}(\mathbf{X})\bz{q}(\mathbf{X})$
and $\mathbf{A}_{i,q}(\mathbf{X})=\gamma_{i}(\mathbf{X})\left(\ek^{q}(\mathbf{X})-\bz{q}(\mathbf{X})\right)=\gamma_{i}(\mathbf{X})\sum_{j=1}^{3}\mathbf{u}_{j,q}(\mathbf{X})$.
Since $\mathbf{X}$ and $\mathbf{Y}$ are independent and the only
part of $B_{i,q}(\mathbf{X})$ that is random is $\mathbf{X},$ then
$Cov\left[B_{1,q}(\mathbf{X}),B_{2,r}(\mathbf{Y})\right]=0.$ Also
$Cov\left[\mathbf{A}_{1,q}(\mathbf{X}),B_{2,r}(\mathbf{Y})\right]=0$
and $Cov\left[B_{1,q}(\mathbf{X}),\mathbf{A}_{2,r}(\mathbf{Y})\right]=0$
by Eq.~\ref{eq:cov2} since neither $E_{0}$ or $E_{1}$ hold in
these cases. This leaves only the $Cov\left[\mathbf{A}_{1,q}(\mathbf{X}),\mathbf{A}_{2,r}(\mathbf{Y})\right]$
term. By applying Eqs.~\ref{eq:cov1} and~\ref{eq:cov2}, we obtain
\[
Cov\left[\gamma_{1}(\mathbf{X})\sum_{j=2}^{3}\mathbf{u}_{j,q}(\mathbf{X}),\gamma_{2}(\mathbf{Y})\sum_{j=1}^{3}\mathbf{u}_{j,r}(\mathbf{Y})\right]=o\left(\frac{1}{M_{1}}+\frac{1}{M_{2}}+\frac{1}{k_{1}^{2}}+\frac{1}{k_{2}^{2}}\right),\]
\[
Cov\left[\gamma_{1}(\mathbf{X})\mathbf{u}_{1,q}(\mathbf{X}),\gamma_{1}(\mathbf{Y})\mathbf{u}_{1,r}(\mathbf{Y})\right]=o\left(\frac{1}{M_{1}}+\frac{1}{M_{2}}+\frac{1}{k_{1}^{2}}+\frac{1}{k_{2}^{2}}\right)\]
\[
+1_{\{q=1,r=1\}}\left(c_{7,1}\left(\frac{\gamma_{1}(x)}{\mathbb{E}_{X}\fhat{2}(x)},\frac{\gamma_{2}(x)}{\mathbb{E}_{X}\fhat{2}(x)}\right)\left(\frac{1}{M_{1}}\right)+c_{7,2}\left(\frac{-\gamma_{1}(x)\mathbb{E}_{X}\fhat{1}(x)}{\mathbb{E}_{X}\fhat{2}(x)},\frac{-\gamma_{2}(x)\mathbb{E}_{X}\fhat{1}(x)}{\mathbb{E}_{X}\fhat{2}(x)}\right)\left(\frac{1}{M_{2}}\right)\right)\]
\[
=1_{\{q=1,r=1\}}\left(c_{8,1}\left(\gamma_{1}(x),\gamma_{2}(x)\right)\left(\frac{1}{M_{1}}\right)+c_{8,2}\left(\gamma_{1}(x),\gamma_{2}(x)\right)\left(\frac{1}{M_{2}}\right)\right)+o\left(\frac{1}{M_{1}}+\frac{1}{M_{2}}+\frac{1}{k_{1}^{2}}+\frac{1}{k_{2}^{2}}\right).\]
Combining these results completes the proof.
\end{proof}
Since $\mathbf{X}_{1}$ and $\mathbf{X}_{2}$ are independent, $\mathbb{E}\left[\mathbf{p}_{1}\left(\mathbf{p}_{2}+\mathbf{q}_{2}+\mathbf{r}_{2}+\mathbf{s}_{2}\right)\right]=0.$
Under assumption $(\mathcal{A}.4$) and applying Lemma~\ref{lem:ekhat},
\[
\mathbb{E}\left[\left(\mathbf{p}_{1}+\mathbf{q}_{1}+\mathbf{r}_{1}+\mathbf{s}_{1}\right)^{2}\right]=\mathbb{E}\left[\mathbf{p}_{1}^{2}\right]+o(1)=\var\left[g\left(\mathbb{E}_{\mathbf{X}_{1}}\lhat\left(\mathbf{X}_{1}\right)\right)\right]+o(1)\]
\begin{equation}
=c_{9}\left(\mathbb{E}_{Z}\lhat\left(z\right)\right)+o(1).\label{eq:pvar}\end{equation}
 Also applying Lemma~\ref{lem:cov_ekhat} gives \begin{eqnarray*}
\mathbb{E}\left[\mathbf{q}_{1}\mathbf{q}_{2}\right] & = & c_{8,1}\left(g'\left(\mathbb{E}_{X}\lhat(x)\right),g'\left(\mathbb{E}_{X}\lhat(x)\right)\right)\left(\frac{1}{M_{1}}\right)+c_{8,2}\left(g'\left(\mathbb{E}_{X}\lhat(x)\right),g'\left(\mathbb{E}_{X}\lhat(x)\right)\right)\left(\frac{1}{M_{2}}\right)\\
 &  & +o\left(\frac{1}{M_{1}}+\frac{1}{M_{2}}+\frac{1}{k_{1}^{2}}+\frac{1}{k_{2}^{2}}\right),\\
\mathbb{E}\left[\mathbf{q}_{1}\mathbf{r}_{2}\right] & = & o\left(\frac{1}{M_{1}}+\frac{1}{M_{2}}+\frac{1}{k_{1}^{2}}+\frac{1}{k_{2}^{2}}\right),\\
\mathbb{E}\left[\mathbf{r}_{1}\mathbf{r}_{2}\right] & = & o\left(\frac{1}{M_{1}}+\frac{1}{M_{2}}+\frac{1}{k_{1}^{2}}+\frac{1}{k_{2}^{2}}\right),\end{eqnarray*}
We use Cauchy-Schwarz and Lemma~\ref{lem:ekhat} to get\begin{align*}
\mathbb{E}\left[g'\left(\mathbb{E}_{\mathbf{X}_{1}}\lhat\left(\mathbf{X}_{1}\right)\right)\ek\left(\mathbf{X}_{1}\right)\Psi\left(\mathbf{X}_{2}\right)\ek^{\lambda}\left(\mathbf{X}_{2}\right)\right]\end{align*}
 \begin{eqnarray*}
 & \leq & \sqrt{\mathbb{E}\left[\Psi^{2}\left(\mathbf{X}_{2}\right)\right]\mathbb{E}\left[\left(g'\left(\mathbb{E}_{\mathbf{X}_{1}}\lhat\left(\mathbf{X}_{1}\right)\right)\ek\left(\mathbf{X}_{1}\right)\right)^{2}\ek^{2\lambda}\left(\mathbf{X}_{2}\right)\right]}\\
 & \leq & \sqrt{\mathbb{E}\left[\Psi^{2}\left(\mathbf{X}_{2}\right)\right]\sqrt{\mathbb{E}\left[\left(g'\left(\mathbb{E}_{\mathbf{X}_{1}}\lhat\left(\mathbf{X}_{1}\right)\right)\ek\left(\mathbf{X}_{1}\right)\right)^{4}\right]\mathbb{E}\left[\ek^{4\lambda}\left(\mathbf{X}_{2}\right)\right]}}\\
 & = & \sqrt{\mathbb{E}\left[\Psi^{2}\left(\mathbf{X}_{2}\right)\right]\sqrt{O\left(\frac{1}{k_{1}^{2}}+\frac{1}{k_{2}^{2}}\right)O\left(\frac{1}{k_{1}^{2\lambda}}+\frac{1}{k_{2}^{2\lambda}}\right)}}\\
 & = & \sqrt{\mathbb{E}\left[\Psi^{2}\left(\mathbf{X}_{2}\right)\right]}o\left(\frac{1}{k_{1}^{\lambda/2}}+\frac{1}{k_{2}^{\lambda/2}}\right).\end{eqnarray*}
Lemma~\ref{lem:bounded} and assumption $\left(\mathcal{A}.5\right)$
implies that $\mathbb{E}\left[\Psi^{2}\left(\mathbf{X}_{2}\right)\right]=O(1)$
and from assumption $(\mathcal{A}.3),$ $o\left(\frac{1}{k_{i}^{\lambda/2}}\right)=o\left(\frac{1}{M_{i}}\right).$
This implies that $\mathbb{E}\left[\mathbf{q}_{1}\mathbf{s}_{2}\right]=o\left(\frac{1}{M_{1}}+\frac{1}{M_{2}}\right).$
Similarly, $\mathbb{E}\left[\mathbf{r}_{1}\mathbf{s}_{2}\right]=o\left(\frac{1}{M_{1}}+\frac{1}{M_{2}}\right)$
and $\mathbb{E}\left[\mathbf{s}_{1}\mathbf{s}_{2}\right]=o\left(\frac{1}{M_{1}}+\frac{1}{M_{2}}\right).$
So finally, \begin{eqnarray*}
\var\left[\gk\right] & = & \frac{c_{9}\left(\mathbb{E}_{X}\lhat\left(x\right)\right)}{N}+\frac{N-1}{N}\mathbb{E}\left[\mathbf{q}_{1}\mathbf{q}_{2}\right]+o\left(\frac{1}{M_{1}}+\frac{1}{M_{2}}+\frac{1}{N}+\frac{1}{k_{1}^{2}}+\frac{1}{k_{2}^{2}}\right)\\
 & = & c_{9}\left(\mathbb{E}_{X}\lhat\left(x\right)\right)\left(\frac{1}{N}\right)+c_{8,1}\left(g'\left(\mathbb{E}_{X}\lhat(x)\right),g'\left(\mathbb{E}_{X}\lhat(x)\right)\right)\left(\frac{1}{M_{1}}\right)\\
 &  & +c_{8,2}\left(g'\left(\mathbb{E}_{X}\lhat(x)\right),g'\left(\mathbb{E}_{X}\lhat(x)\right)\right)\left(\frac{1}{M_{2}}\right)+o\left(\frac{1}{M_{1}}+\frac{1}{M_{2}}+\frac{1}{N}+\frac{1}{k_{1}^{2}}+\frac{1}{k_{2}^{2}}\right)\\
 & = & c_{9}\left(L\left(x\right)\right)\left(\frac{1}{N}\right)+c_{8,1}\left(g'\left(L(x)\right),g'\left(L(x)\right)\right)\left(\frac{1}{M_{1}}\right)+c_{8,2}\left(g'\left(L(x)\right),g'\left(L(x)\right)\right)\left(\frac{1}{M_{2}}\right)\\
 &  & +o\left(\frac{1}{M_{1}}+\frac{1}{M_{2}}+\frac{1}{N}+\frac{1}{k_{1}^{2}}+\frac{1}{k_{2}^{2}}\right)\\
 & = & c_{9}\left(\frac{1}{N}\right)+c_{8,1}\left(\frac{1}{M_{1}}\right)+c_{8,2}\left(\frac{1}{M_{2}}\right)+o\left(\frac{1}{M_{1}}+\frac{1}{M_{2}}+\frac{1}{N}+\frac{1}{k_{1}^{2}}+\frac{1}{k_{2}^{2}}\right),\end{eqnarray*}
 where the second to last step follows from $\mathbb{E}_{X}\lhat(X)=L(X)+o(1).$

\bibliographystyle{ieeetr}
\bibliography{divergence_isit2014}

\end{document}